\documentclass[11pt]{article}
\usepackage[margin=26mm]{geometry}
\usepackage{amsthm}
\usepackage{amsmath}
\usepackage{amsfonts}
\usepackage{amssymb}
\usepackage{graphicx}
\usepackage{algorithm}
\usepackage{url}
\usepackage{epsfig, subfigure}
\usepackage{wrapfig}
\usepackage{color}
\usepackage{subfig}
\usepackage{algorithm}
\usepackage{algpseudocode}
\usepackage{comment}
\usepackage{lineno}
\usepackage[T1]{fontenc}
\usepackage[utf8]{inputenc}
\usepackage{authblk}

\newcommand{\R}{\mathbb{R}}

\newcommand{\vol}{\mathrm{vol}}
\newcommand{\ball}{B}
\newcommand{\union}{\mathrm{union}}
\newcommand{\overlap}{\mathrm{overlap}}
\newcommand{\distoverlap}{\mathrm{dist\_overlap}}
\newcommand{\voloverlap}{\mathrm{vol\_overlap}}
\newcommand{\multiplicity}{\mathrm{multiplicity}}
\newcommand{\density}{\mathrm{density}}

\newcommand{\lattice}{\mathcal{L}}
\newcommand{\voronoi}{V}
\newcommand{\threshold}{\omega}
\newcommand{\packingquality}{\mathrm{Qual}_{\mathrm{packing}}}
\newcommand{\coveringquality}{\mathrm{Qual}_{\mathrm{covering}}}
\newcommand{\freespace}{\mathrm{free\_space}}

\newcommand{\ignore}[1]{}

\providecommand{\norm}[1]{\lVert#1\rVert}

\newtheorem{theorem}{Theorem}
\newtheorem{lemma}[theorem]{Lemma}
\newtheorem{proposition}[theorem]{Proposition}

\title{Sphere Packing with Limited Overlap}

\author[1]{Mabel Iglesias-Ham\thanks{mabel.iglesias-ham@ist.ac.at}}
\author[2]{Michael Kerber\thanks{mkerber@mpi-inf.mpg.de}}
\author[1]{Caroline Uhler\thanks{caroline.uhler@ist.ac.at}}
\affil[1]{Institute of Science and Technology Austria, Klosterneuburg}
\affil[2]{Max-Planck-Institut f\"ur Informatik, Saarbr\"ucken}

\date{\vspace{-5ex}}

\begin{document}
\maketitle

\begin{abstract}
The classical sphere packing problem asks for the best (infinite) arrangement of non-overlapping unit balls which cover as much space as possible.
We define a generalized version of the problem, where we allow each ball a limited amount of overlap with other balls.
We study two natural choices of overlap measures
and obtain the optimal lattice packings in a parameterized family of lattices which contains
the FCC, BCC, and integer lattice.
\end{abstract}

\section{Introduction}
\label{intro}

Sphere packing and sphere covering problems have been a popular area of study in discrete mathematics over many years. A sphere packing usually refers to the arrangement of non-overlapping \mbox{n-d}imensional spheres. A typical sphere packing problem is to find a maximal density arrangement, i.e.,~an arrangement in which the spheres fill as much of the space as possible. On the other hand, sphere covering refers to an arrangement of spheres that cover the whole space. Overlap is not only allowed in these arrangements, but inevitable. In this case, the aim is to find an arrangement that minimizes the density (i.e.,~the total volume of the spheres divided by the volume of the space). 

In dimension 2, the densest circle packing and the thinnest circle covering are both attained by the hexagonal lattice~\cite{hexPacking1939}. In dimension 3, Hales~\cite{Hales2005} has recently given a computer-assisted proof showing that the face-centered cubic (FCC) lattice achieves the densest packing even when the sphere centers are not constrained to lie on a lattice. The thinnest covering in dimension 3 is achieved by the body-centered cubic (BCC) lattice~\cite{Bambah1954}, but it is not known yet if one can improve the covering by allowing non-lattice arrangements. In dimension 4 and higher, the situation is more complicated and even less is known; 
see~\cite{cs-sphere} for a comprehensive summary.

Although sphere packing and sphere covering problems have attracted a lot of attention by mathematicians, the arrangements of spheres encountered for example in modeling in the biological sciences usually fall between sphere packing and sphere covering: Models consist of overlapping spheres, which do not fill the whole space, and one is often interested in maximal density configurations of spheres where we allow a certain amount of overlap. Examples are the spatial organization of chromosomes in the cell nucleus~\cite{chromosome_territories, ellipsoid_packing}, the spatial organization of neurons~\cite{neuron_packing_1, neuron_packing_2}, or the arrangement of ganglion cell receptive fields on the retinal surface~\cite{rfield_packing_1, rfield_packing_2}. The wide applicability is also based on the fact that soft spheres can be modeled as hard spheres with limited overlap. In all these applications one would like to understand the optimal packing configuration of spheres when allowing a certain amount of overlap. 

In this paper, we study this problem between sphere packing and sphere covering for the special case when the sphere centers lie on a particular family of lattices obtained by diagonally distorting the integer grid: Let $\delta>0$ be a distortion parameter. Then the lattice $\mathcal{L}_{\delta}$ is defined by mapping each unit vector $e_i\in\mathbb{R}^n$, $i=1,\dots ,n$, to 
\begin{equation}
e_i^{\delta} \;:=\; e_i +\frac{\delta -1}{n}\mathbf{1}.
\label{dist_family}
\end{equation}
This family of lattices has been defined and studied in \cite{ek-covering}. It is particularly interesting, since it contains the optimal packing lattices in dimensions 2 and 3 and the optimal covering lattices in dimensions 2-5. At the same time, it is simple enough (defined by one parameter only) allowing us to give a complete analysis of the density of sphere arrangements with limited overlap as a function of $\delta$. In this paper, we prove that for dimension 2 and 3, the optimal packing and covering lattices are robust: even when allowing a certain overlap, either the optimal sphere packing lattice or the optimal sphere covering lattice attain the maximum density, depending on the amount of allowed overlap and how overlap is measured.  

Our paper is organized as follows:
In Section \ref{measures} we discuss two different measures of overlaps in sphere arrangements. The first one, called \emph{distance-based overlap}, is simply a linear function of the distance between two sphere centers 
and has been used for the analysis in \cite{ellipsoid_packing}. The second one, called \emph{volume-based overlap} is based on the intersection volume of spheres. 
In Section \ref{dist_overlap} we give a complete description of the density of sphere arrangements with limited overlap for the distance-based overlap. In particular, we show that the FCC lattice results in the densest arrangement in the considered family of lattices, regardless of the amount of allowed overlap.
In Section \ref{vol_overlap} we analyze the more complicated volume-based overlap measure: We derive an exact formula for the packing density for each lattice in the family and each overlap threshold
by analyzing the Voronoi polytope of a lattice point. For planar lattices, we prove that the hexagonal lattice remains optimal for any overlap. In dimension 3, we show that the best choice depends
on the allowed overlap and we provide numerical evidence that the optimal lattice is always either the FCC or the BCC lattice.
We end with a short discussion in Section \ref{discussion}.

\section{Measures of Sphere Arrangements}
\label{measures}

We let $\ball_r(p)$ denote the closed ball of radius $r$ and center $p$.
We fix a lattice $\lattice$ and let $\voronoi$ denote the Voronoi cell of the origin consisting of all points that are closer to the origin than to any other lattice point. Observe that the Voronoi cells
of other lattice points are just translations of $\voronoi$ and that the Voronoi cells tessellate $\mathbb{R}^n$. 

A first measure of a sphere arrangement is the \emph{density}. It is defined as the number of spheres that contain an average point and can be rephrased as
\begin{equation}
\density_{\lattice}(r)\;:=\; \frac{\vol\, B_r(0)}{\vol\,\voronoi}.
\label{density}
\end{equation}

Second, we define the \emph{union} of a sphere arrangement to be
\begin{equation}
\union_{\lattice}(r)\;:=\;\frac{\vol \left(B_r(0)\cap \voronoi\right)}{\vol\,\voronoi}.
\label{union}
\end{equation}

The $\union$ denotes what fraction of the Voronoi cell is covered
by the ball of radius $r$. Looking at the whole space, it also
denotes what fraction of $\R^n$ is covered by the union of all balls
of radius $r$. This follows because the Voronoi cells tessellate
$\R^n$ and from the following statement:
\begin{proposition}
Let $p$ be a point that belongs to the Voronoi cell of $c_1$. 
If $p$ is covered by a ball $B_r(c_2)$, then $p$ is also covered by $B_r(c_1)$.
\end{proposition}

A third measure of a sphere arrangement is the \emph{overlap}. We define two measures of overlap.  The \emph{distance-based overlap} was used to model the spatial organization of chromosomes in \cite{ellipsoid_packing} and is defined as the diameter of the largest sphere that can be inscribed into the intersection of two spheres, i.e.:
\begin{equation}
\distoverlap_{\lattice}(r)\;:=\; \max\left(\frac{2r - \min_{\ell\in\mathcal{L}\setminus\{0\}}(\norm{\ell})}{2r},0\right).
\label{distoverlap}
\end{equation}
A less simplified measure of overlap is the \emph{volume-based overlap}, which we define as the fraction of a sphere that expands outside its Voronoi cell:
\begin{equation}
\voloverlap_{\lattice}(r):= \frac{\vol\, B_r(0) - \vol (B_r(0)\cap \vol\,\voronoi)}{\vol\,\voronoi}.
\label{voloverlap}
\end{equation}
This value is equivalent to the fraction of all other spheres expanding into a Voronoi cell (i.e.,~the overlap with multiplicity inside a Voronoi cell).

We observe that $\density_{\lattice}(\cdot)$, $\union_{\lattice}(\cdot)$ and $\overlap_{\lattice}(\cdot)$ (describing both overlap measures) are non-negative, monotonously increasing functions with $\union_{\lattice}(\cdot)$ upper bounded by~1. The upper bound for $\union$ is reached exactly at the covering radius, the maximal distance of the origin to the boundary of $\voronoi$. The lower bound for $\overlap$ is reached exactly at the packing radius, the minimal distance of the origin to the boundary of $\voronoi$. Also, it holds that
\begin{eqnarray}
\label{density-overlap-union}
\voloverlap_{\lattice}(r) \;=\; \density_{\lattice}(r) - \union_{\lattice}(r).
\end{eqnarray}

\ignore{
Finally, a forth way of quantifying a sphere arrangement is given by its \emph{multiplicity}, which we define as:
$$\multiplicity_{\lattice}(r)\;:=\; \frac{\vol\, B_r(0) }{\vol (B_r(0)\cap \vol\,\voronoi)}.$$
Note that a $\multiplicity$ larger than 1 corresponds to the presence of overlap. Notice also that the multiplicity of a sphere arrangement is related to density and union by the following equation:
$$\density_{\lattice}(r)\; :=\; \union_{\lattice}(r)\cdot\multiplicity_{\lattice}(r),$$
because the intersection of $\voronoi$ and $B_r(0)$ cancels out. Moreover, if we constrain the union to be 1,  the intersection in the denominator becomes $\vol\, \voronoi$, and we have the density. If we constrain the multiplicity to be 1 (meaning that only non-overlapping configurations are allowed), the intersection simplifies to $\vol\, B_r(0)$, and we have the density as well. So using the multiplicity it is possible to ``interpolate'' between the classical packing and covering problems. 
}

\vspace{0.3cm}

Building upon these measures of sphere arrangements we can now define a relaxed packing and covering quality when allowing overlap and uncovered space, respectively. By fixing a threshold $\threshold\in \mathbb{R}_{\geq 0}$, we define the 
\emph{relaxed packing quality} of a lattice as
$$\packingquality(\lattice, \omega)\;:=\;\max_{r\geq 0} \left\{ \density_{\lattice}(r)\mid\overlap_{\lattice}(r)\leq\threshold\right\}.$$

The goal is to find the lattice that maximizes $\packingquality$. Note that for $\threshold=0$, this is equivalent to the classical sphere
packing problem: We want to cover as much space as possible by balls without overlap. It is known that in dimension 3 the FCC lattice is the optimal solution to this problem.

\begin{lemma}
\label{lemma_packing}
The FCC lattice is not optimal with respect to $\packingquality$ for all values of $\threshold$ when measuring overlap by $\voloverlap$.
\end{lemma}
\begin{proof}
Let $\omega$ be the overlap of the BCC lattice when choosing the radius to be its covering radius. Note that the density of this covering is $1+\omega$ by~\eqref{density-overlap-union}. 
Assume that the FCC lattice attains the same density for $\omega$. Then, again by~\eqref{density-overlap-union}, the union must be $1$, so the FCC lattice yields a sphere covering with the same density as the BCC lattice.
But this is a contradiction to the well-known fact that the FCC covering density is strictly larger than the BCC covering density.
\end{proof}

Interestingly, we will prove in Section \ref{dist_overlap} that the FCC lattice is in fact optimal for all values of $\threshold$ when measuring overlap by $\distoverlap$. Similarly, we can define a \emph{relaxed covering quality} as
$$\coveringquality(\lattice, \omega)\;:=\;\min_{r\geq 0} \left\{ \density_{\lattice}(r)\mid 1-\union_{\lattice}(r)\leq\threshold\right\}.$$

In words, we want as little overlap as possible while allowing only a certain
amount of uncovered space. Note that for $\threshold=0$, this is equivalent to the classical 
covering problem: We want to cover the whole space by balls minimizing the density. Similarly as in Lemma \ref{lemma_packing} we can prove that the BCC lattice is not optimal with respect to $\coveringquality$ for all values of $\threshold$ when measuring overlap by $\voloverlap$. However, the BCC lattice is optimal for all values of $\threshold$ when measuring overlap by $\distoverlap$ as we will see in Section \ref{dist_overlap}.

Since for the applications we described in the introduction the relevant quality measure seems to be $\packingquality$, we will mainly concentrate on this measure. However, the analysis could easily be extended to $\coveringquality$.

\vspace{0.4cm}

From now on, we focus on the lattices given by a diagonal distortion of the integer lattice in $\R^n$ as defined in~\eqref{dist_family}.
The parameter $\delta$ defines the amount of distortion, with $\delta=1$ denoting no distortion. For $\delta$ from $1$ to $0$, every point of the integer lattice undergoes a continuous motion towards its projection
onto the plane with normal vector $(1,\ldots,1)$. For $\delta\geq 1$, each lattice point moves continuously in the opposite direction. For $n=2$, the hexagonal lattice corresponds to $\delta=1/\sqrt{3}$ and $\delta=\sqrt{3}$,
and for $n=3$, the FCC lattice corresponds to $\delta=2$ and the BCC lattice to $\delta=1/2$; see~\cite{ek-covering} for more details.

To simplify notation, we write e.g.~$\density(\delta , r)$ instead of $\density_{\lattice_{\delta}}(r)$. Fixing a threshold $\threshold$ for the overlap, we would like to find the best lattice in the family such that $\packingquality(\delta,\omega)$ is maximized. The approach we take is to compute $\packingquality$ for a given $\delta$ in two steps:
\begin{enumerate}
\item Compute the largest ball radius $r(\delta, \omega)$ 
      such that $\overlap(\delta, r(\delta,\omega))\leq\threshold$.
\item Compute $\packingquality(\delta,\omega)=\density(\delta, r(\delta, \omega))$.
\end{enumerate}

\section{Distance-based overlap}
\label{dist_overlap}

In \cite{ellipsoid_packing}, an algorithm was developed for finding minimum overlap configurations of $N$ spheres (or more generally ellipsoids) packed into an ellipsoidal container. In order to get an efficient algorithm, the simplified distance-based overlap measure was used, which could be computed as a convex optimization problem. One can easily check that the problem of finding minimal overlap configurations of spheres with a certain density is equivalent to finding maximal density configurations of spheres with a certain overlap, the problem we study in this paper. It was observed in a few examples (see Example 3.4 in \cite{ellipsoid_packing}) that the optimal configuration of the spheres is invariant to scaling of the radii. This is in fact an important property for the application to chromosome packing, since the exact chromatin packing density is not known and one would hope that the positioning is robust to different scalings of the chromosomes. In the following, we prove that this scaling-invariance holds in infinite space when the sphere centers are restricted to lie on the 1-parameter distortion family of the integer grid.

For the 1-parameter family of diagonal distortions defined in (\ref{dist_family}), the density simplifies to

\begin{equation}
\density(\delta,r) = \frac{V_n r^n}{\delta},
\label{def_density}
\end{equation}

where $V_n$ denotes the (n-dimensional) volume of the n-dimensional unit ball, i.e.

$$V_n = \left\{ \begin{array}{ll}
\pi^{\frac{n}{2}}/\left(\frac{n}{2}\right)! & \quad\textrm{if $n$ is even},\\
\pi^{\frac{n-1}{2}}2^{\frac{n+1}{2}}/n!! & \quad\textrm{if $n$ is odd}.
  \end{array} \right. $$

Here $n!! = n \cdot(n-2) \cdot \dots \cdot 3 \cdot 1$ denotes the double factorial. Half the minimal distance between two lattice points, $\min_{p \in\partial{V}} \|p\|$, has been computed in \cite{ek-covering} (it corresponds to the packing radius):

\begin{equation}
\min_{p \in\partial{V}} \|p\| = \left\{ \begin{array}{ll}
\frac{1}{2} \delta \sqrt{n} & \quad\textrm{for $0 \leq \delta \leq \frac{1}{\sqrt{n+1}}$},\\
\frac{1}{2} \sqrt{1+\frac{\delta^2-1}{n}} & \quad\textrm{for $\frac{1}{\sqrt{n+1}} \leq \delta \leq \sqrt{n+1}$},\\
\frac{1}{2} \sqrt{2}, & \quad\textrm{for $\sqrt{n+1} \leq \delta$}.
  \end{array} \right.
  \label{packing_radius}
  \end{equation}
  
Using these formulas we now prove that the maximum density configuration does not depend on the amount of allowed overlap and is always attained by $\delta = \sqrt{n+1}$, which corresponds to the optimal packing lattice in the family for all $n\geq 2$ and over all lattices in dimension 2 and 3.

\begin{theorem}
\label{thm_overlap}
The lattice $\lattice_\delta$ which maximizes the relaxed packing quality w.r.t.~$\distoverlap$, i.e.
\begin{align*}
\max_{\delta >0, r\geq 0} \qquad& \density(\delta,r)
\\
\mbox{subject to} \qquad &
\distoverlap(\delta,r) \leq \omega
\end{align*}
is attained by $\delta = \sqrt{n+1}$ independent of the value of $\omega\in[0,1)$. 
\end{theorem}

\begin{proof}
By plugging the packing radius given in (\ref{packing_radius}) into the definition of $\distoverlap$ in (\ref{distoverlap}), we can solve for $r(\delta,\omega)$\footnote{Note that $ \min_{\ell\in\mathcal{L}\setminus\{0\}}(\norm{\ell}) = 2 \cdot \min_{p\in\partial{V}} \|p\|$}. Then plugging $r(\delta,\omega)$ into the formula for the density given in (\ref{density}) we get
\begin{equation}
\density(\delta,r(\delta,\omega)) = \left\{ \begin{array}{ll}
\frac{n^{\frac{n}{2}}V_n}{2^n(1-\omega)^n}\, \delta^{n-1} & \quad\textrm{for $0 \leq \delta \leq \frac{1}{\sqrt{n+1}}$},\\
\frac{V_n}{2^n(1-\omega)^n}\, \delta^{-1}\left(1+\frac{\delta^2-1}{n}\right)^{\frac{n}{2}} & \quad\textrm{for $\frac{1}{\sqrt{n+1}} \leq \delta \leq \sqrt{n+1}$},\\
\frac{V_n}{2^{\frac{n}{2}}(1-\omega)^n}\, \delta^{-1} & \quad\textrm{for $\sqrt{n+1} \leq \delta$}.\\
  \end{array} \right. 
  \label{eq_density}
  \end{equation}
The function $\density(\delta,r(\delta,\omega))$ for $n=3$ and $\omega=0.5$ is shown in Figure \ref{fig_overlap_dual_plot} (left). Since $\omega<1$, the constants in the function $\density(\delta,r(\delta,\omega))$ in (\ref{eq_density}) are positive. By taking derivatives with respect to $\delta$ we find that for $0 < \delta \leq 1/\sqrt{n+1}$ the density is strictly increasing for all values of $\omega$. Similarly, for the branch $1/\sqrt{n+1} \leq \delta \leq \sqrt{n+1}$ the density is strictly decreasing for $\delta<1$, achieves a minimum at $\delta = 1$, and is strictly increasing for  $\delta>1$, independent of the value of $\omega$. Finally, for $\delta \geq \sqrt{n+1}$ the density is strictly decreasing for all values of $\omega$. As a consequence,
 $$\max_{\delta>0}\; \density\!\left(\delta,r(\delta,\omega)\right) \; = \; \max\!\left(\!\density\!\left(\frac{1}{\sqrt{n+1}}, r\left(\frac{1}{\sqrt{n+1}},\omega\right)\!\!\right), \density\!\left(\!\sqrt{n+1},  r(\sqrt{n+1},\omega)\right)\!\!\right)\!.$$
Since $n^{\frac{n}{2}}(n+1)^{-\frac{n-1}{2}} \leq 2^{\frac{n}{2}}(n+1)^{-\frac{1}{2}}$ for all $n\geq 2$ with equality if and only if $n=2$, the maximum is attained by $\delta = \sqrt{n+1}$ with equality if and only if $n=2$, where both lattices correspond to the hexagonal lattice. 
\end{proof}

\begin{figure}[!bt]
\centering
\includegraphics[width=7cm]{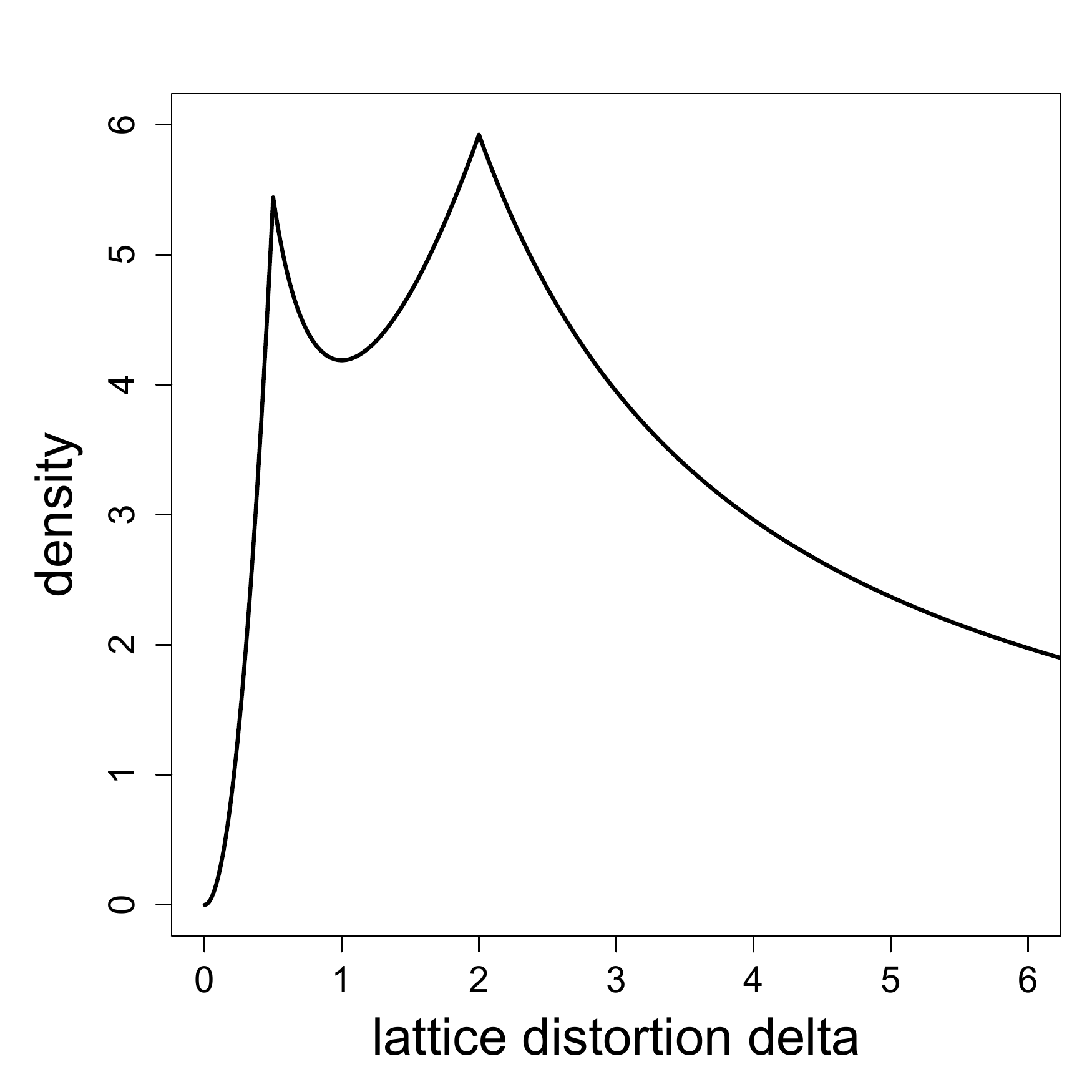} \qquad\qquad \includegraphics[width=7cm]{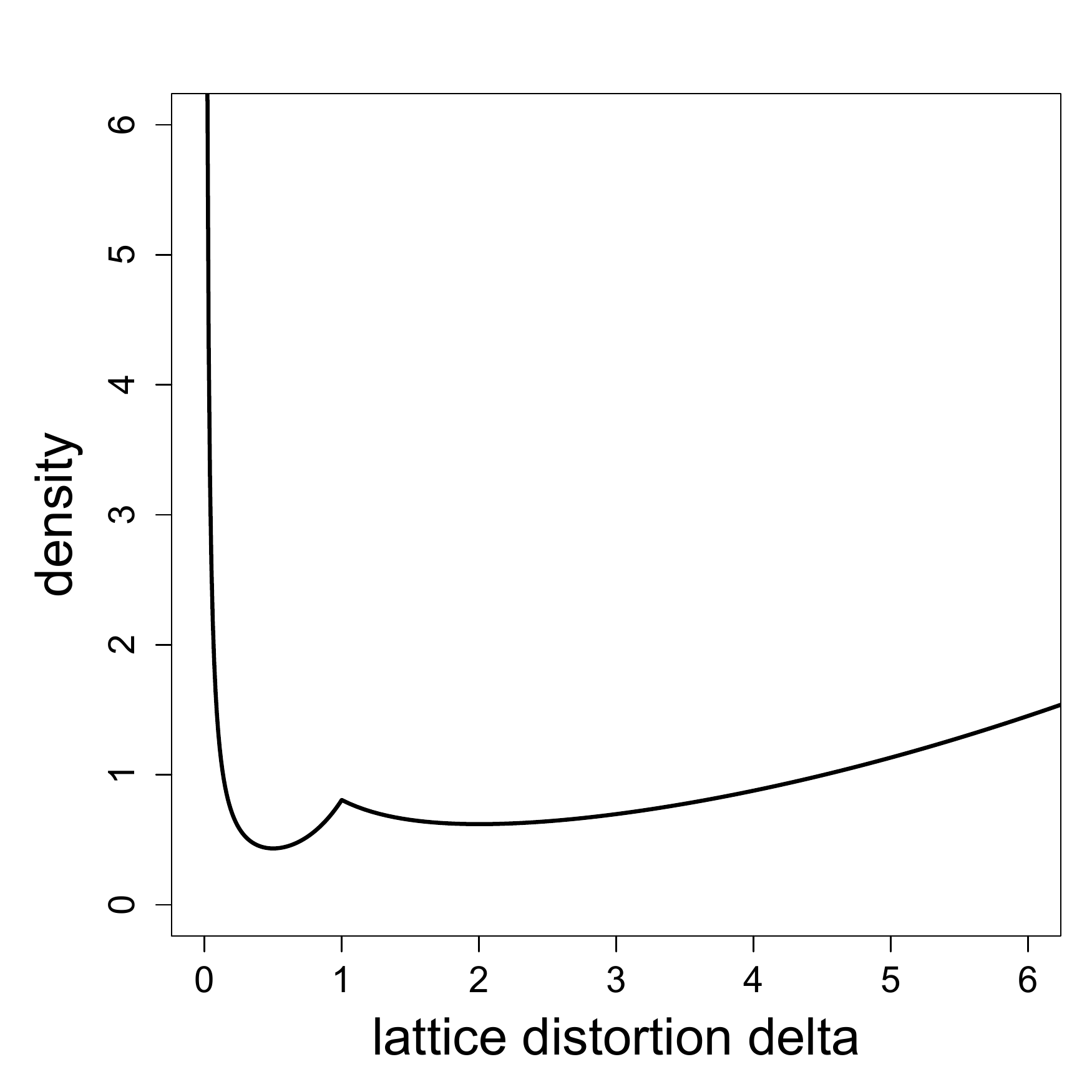}
\caption{$\packingquality$ (left) and $\coveringquality$ (right) as a function of the distortion parameter $\delta$ for $n=3$ and $\omega=0.5$.}
\label{fig_overlap_dual_plot}
\end{figure}

This proves that the sphere configuration which maximizes the density when allowing a certain overlap (measured by the distance-based overlap) is identical to the optimal packing configuration independent of the allowed overlap. We next briefly analyze the relaxed covering quality and show that in this case the optimum is always attained by the optimal covering configuration. Similarly as for $\distoverlap$, we use a linearized measure of the uncovered space $1-\union$. We define it as the largest diameter of a sphere which can be inscribed into the free-space, i.e.:

\begin{equation}
\label{eq_freespace_eta}
\freespace(\delta, r) = \max\left(\frac{\max_{p \in\partial{V}} \|p\|-r}{r},0\right)
\end{equation}

The $\max_{p \in\partial{V}} \|p\|$ has been computed in \cite{ek-covering} for the 1-parameter family of lattices under consideration (it corresponds to the covering radius):

\begin{equation}
\max_{p \in\partial{V}} \|p\| = \left\{ \begin{array}{ll}
\frac{\sqrt{n^2-1+(n^2+2)\delta^2+(n^2-1)\delta^4}}{\sqrt{12n}} & \quad\textrm{for $0 \leq \delta\leq 1$},\\
\frac{\sqrt{n^2-1+\delta^2}}{2 \sqrt{n}} & \quad\textrm{for $1 \leq \delta$ and $n$ odd},\\
\frac{\sqrt{n^2-2+\delta^2+\frac{1}{\delta^2}}}{2 \sqrt{n}} & \quad\textrm{for $1 \leq \delta$ and $n$ even}.
\end{array} \right.
\label{covering_radius}
\end{equation}

Using these formulas we can show that the maximum density configuration does not depend on the amount of allowed free-space and is always attained by $\delta = 1/\sqrt{n+1}$, which corresponds to the optimal covering lattice in the family for all $n\geq 2$ and over all lattices in dimension 2-5.

\begin{theorem}
\label{thm_freespace}
The lattice $\lattice_\delta$ which minimizes the relaxed covering quality, i.e.
\begin{align*}
\min_{\delta >0, r\geq 0} \qquad& \density(\delta,r)
\\
\mbox{subject to} \qquad &
\freespace(\delta,r) \leq \omega
\end{align*}
is attained by $\delta =1/ \sqrt{n+1}$ independent of the value of $\omega\in\mathbb{R}_{\geq 0}$. 
\end{theorem}

\begin{proof} The proof is analogous to the proof of Theorem \ref{thm_overlap}. The function $\density(\delta,r(\delta,\omega))$ for $n=3$ and $\omega=0.5$ is shown in Figure \ref{fig_overlap_dual_plot} (right).
\end{proof}

\section{Volume-based overlap}
\label{vol_overlap}

In this section we analyze $\packingquality$ for the volume-based overlap measure in dimension 2 and 3. As we have already pointed out in Lemma~\ref{lemma_packing}, we cannot expect the same behavior as for the distance-based overlap measure discussed in the previous section. 
However, the lemma only states that the FCC lattice, which is optimal for $\threshold=0$ is worse than the BCC lattice for some value of $\threshold$. This does not rule out the possibility
of other lattices being optimal. This section will perform a deeper investigation of the optimal lattice configurations, starting with the two-dimensional case. 

\subsection{Dimension 2}
\label{sec_dim_2}

First of all, note that in dimension 2, the lattice for $\delta$ is a scaled version of the lattice for $\frac{1}{\delta}$. Because of this symmetry, it suffices to study all lattices with $0 < \delta\leq 1$.

For analyzing the volume-based overlap measure, we first derive a formula for the volume of $V\cap B_r$. This requires
the investigation of the Voronoi cell $V$ in some detail. $V$ is bounded by six bisectors: four of them
with the lattice points $\pm e_1^{(\delta)}, \pm e_2^{(\delta)}$,
and two with the lattice points $\pm(e_1^{(\delta)}+e_2^{(\delta)})$. See Figure~\ref{fig:Voronoi_cells_delta_smaller_one} for an illustration.
We call the bisectors of type $1$ and type $2$, respectively.
Their distances to the origin are given by $r_1$ and $r_2$, respectively, with
$$r_1\;:=\;\frac{\sqrt{\delta^2+1}}{2\sqrt{2}}, \qquad r_2\;:=\;\frac{\delta}{\sqrt{2}}.$$
Note that $r_1>r_2$ if and only if $\delta<\sqrt{\frac{1}{3}}$.

There are six boundary vertices of $V$ and they all have the same distance to the origin, namely
$$r_3\;:=\;\frac{\delta^2+1}{2\sqrt{2}},$$
which agrees with the covering radius computed in~\cite{ek-covering} (see also (\ref{covering_radius})).
As expected, $r_3\geq\max\{r_1,r_2\}$, with equality if and only if $\delta=1$.

\begin{figure}[!tb]
\centering
\includegraphics[width=4cm]{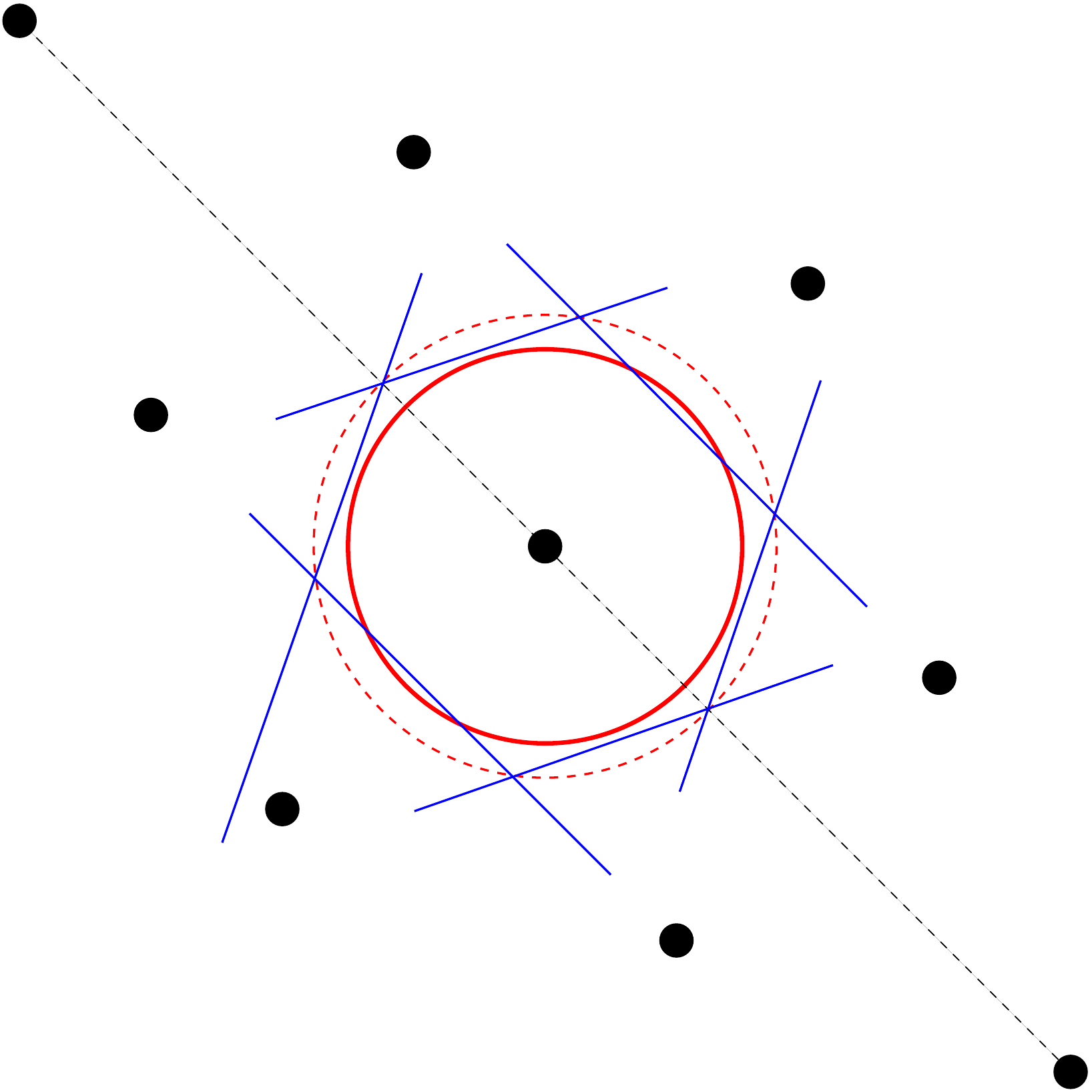}
\hspace{1.2cm}
\includegraphics[width=4cm]{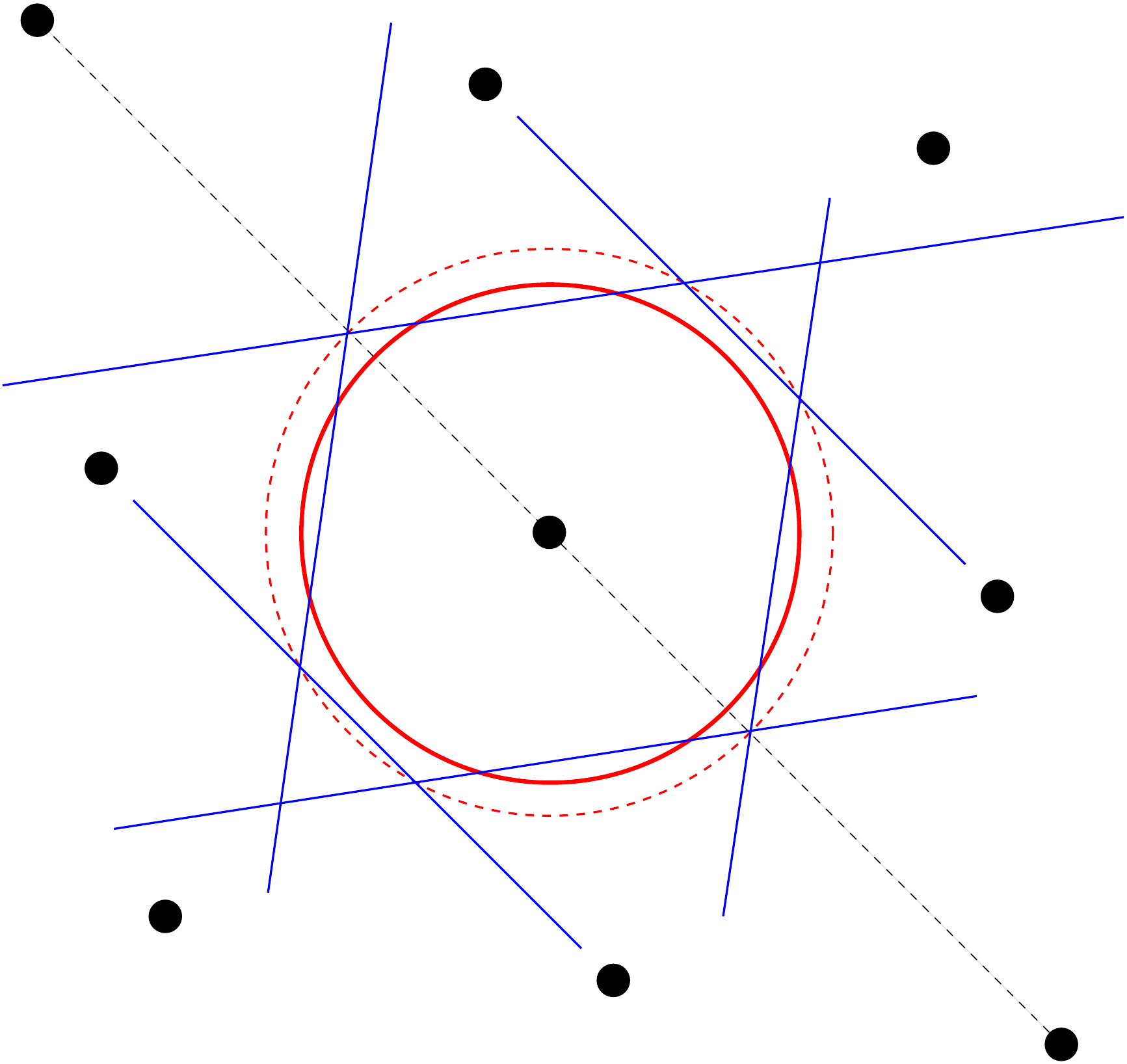}
\hspace{1.5cm}
\includegraphics[width=4cm]{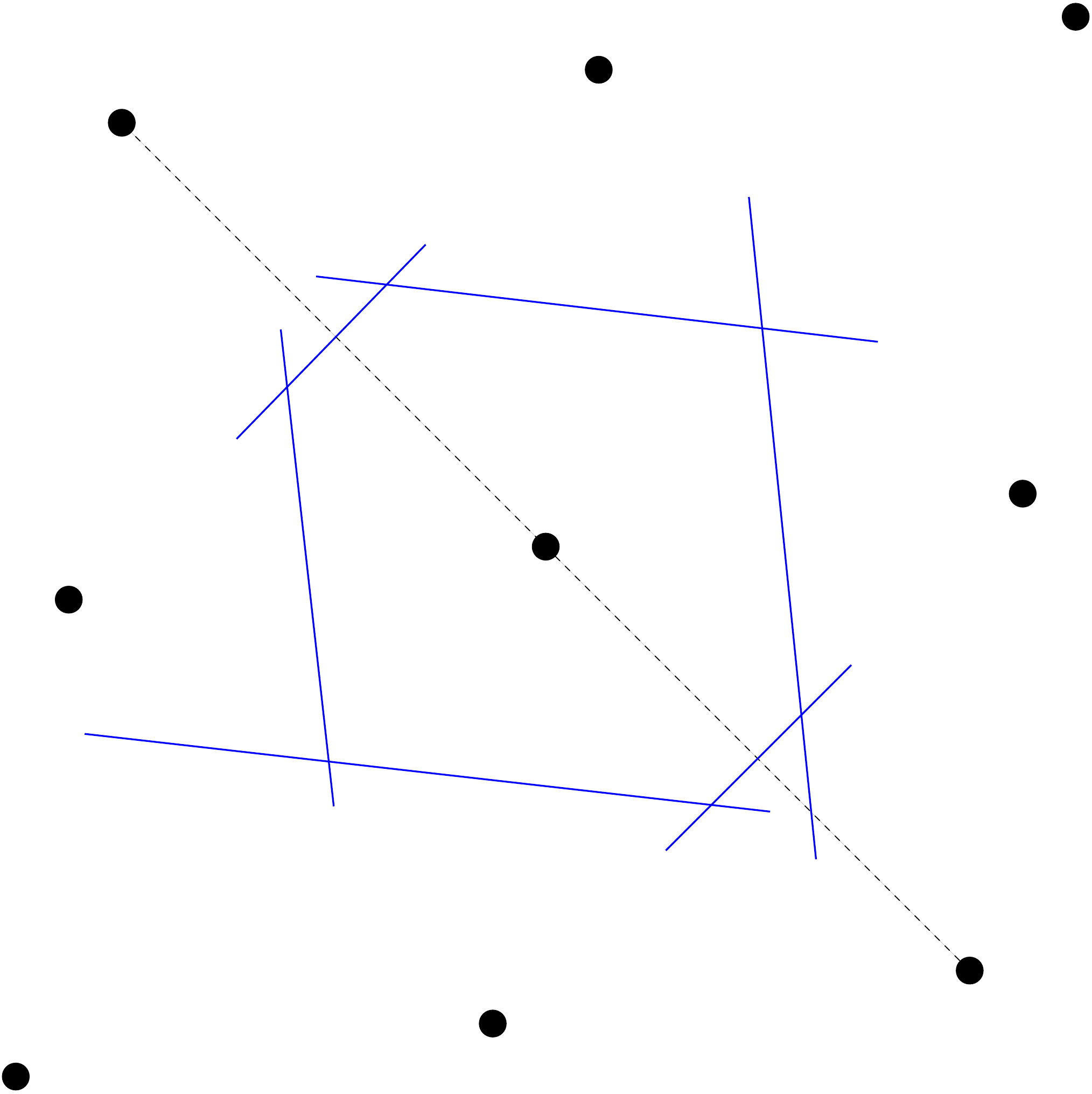}
\caption{The Voronoi cell $V$ for two different values of $\delta<1$ (left and middle) and $\delta>1$ (right). On the left, the bisectors of type $2$ are hit first,
whereas in the middle bisectors of type $1$ are hit first. Note that all lattice points neighboring the origin lie on a common circle around the origin.}
\label{fig:Voronoi_cells_delta_smaller_one}
\end{figure}

With this data we can derive a formula for the volume of $V\cap B_r$: If $r\leq \min\{r_1,r_2\}$, $B_r$ is completely contained
in $V$ and the volume equals the volume of $B_r$. If $r\geq r_3$,
$B_r$ contains all boundary vertices of $V$ and thus all of $V$ (of volume $\delta$), by convexity. In the last case where $\min\{r_1,r_2\} < r < r_3$, the part of $B_r$ that is
not in $V$ is the union of up to six circular segments. Their area is given by
$$A=\frac{r^2}{2}(\Theta-\sin\Theta),$$
where $\Theta$ is the angle at the origin induced by the chord. 
This angle can be expressed as
$$\Theta=2\arccos\left(\frac{d}{r}\right),$$
where $d$ is the smallest distance of the chord to the origin.

In our case, the chord is given by a bisector. Depending on the type of the bisector, 
$d$ is either equal to $r_1$ or equal to $r_2$.
So we define
$$\Theta_1:=\begin{cases}
0 & r<r_1,\\
2\arccos\left(\frac{r_1}{r}\right) & r\geq r_1,
\end{cases}\qquad \Theta_2:=\begin{cases}
0 & r<r_2,\\
2\arccos\left(\frac{r_2}{r}\right) & r\geq r_2.
\end{cases}$$
Since the circular segments do not intersect for any $r<r_3$
(because an intersection would imply that a boundary vertex of $V$
is part of $B_r$) and there are four bisectors of type $1$ and two bisectors of type $2$,
it follows that

$$\vol (V\cap B_r)=\begin{cases}
\pi r^2 & 0\leq r\leq\min\{r_1,r_2\},\;\,\, 0<\delta\leq 1,\\
r^2(\pi-2\Theta_1-\Theta_2+2\sin\Theta_1+\sin\Theta_2) & \min\{r_1,r_2\} \leq r \leq r_3, \; 0<\delta\leq 1,\\
\delta & r_3\leq r,\;\qquad\qquad\qquad\, 0<\delta\leq 1.
\end{cases}$$
Using this formula, we can now prove the following theorem:

\begin{theorem}
In dimension 2, the lattice $\lattice_\delta$ which maximizes the relaxed packing quality w.r.t. $\voloverlap$, i.e.
\begin{align*}
\max_{\delta >0, \;r\geq 0} \qquad& \density(\delta,r)
\\
\mbox{subject to} \qquad &
\voloverlap(\delta,r) \leq \omega
\end{align*}
is attained by the hexagonal lattice (i.e.~$\delta \in\{1/\sqrt{3},\sqrt{3}\}$) independent of the value of $\omega\in\mathbb{R}_{\geq 0}$. 
\end{theorem}

\begin{proof}
Let $\,\threshold\,$ and $\,\delta\,$ be fixed. 
Our goal is to compute 
$\density(\delta,r)$ where $r:=r(\delta,\threshold)$ is chosen maximally such that
$\voloverlap(\delta,r)\leq\threshold$. Observe that the maximal $r$ is certainly
at least the packing radius $\min\{r_1,r_2\}$. This results in the packing density, which is maximized by the hexagonal lattice. Moreover, if $\threshold$ is sufficiently large to allow a covering, i.e. $\omega\geq\voloverlap(\delta,r_3)$, the maximal density is attained at the best covering. This is known to be the hexagonal lattice. So we can concentrate on the case $\min\{r_1,r_2\}\leq r\leq r_3$ where
\begin{equation}
0\leq \voloverlap (\delta, r) \leq \voloverlap (\delta, r_3) =  \frac{\pi(\delta^2+1)^2}{8\delta}-1,
\label{eq_voloverlap_2}
\end{equation}
meaning that
\begin{equation}
\voloverlap (\delta, r)=\frac{r^2}{\delta}(2\Theta_1+\Theta_2-2\sin\Theta_1-\sin\Theta_2).
\label{eq_voloverlap}
\end{equation}

Consider the function $F(\delta,\threshold,r):=\omega-\voloverlap(\delta,r)$,
which is defined for $(\threshold, \delta, r)$ in the limits of interest given in \eqref{eq_voloverlap_2}. By definition, $r=r(\delta,\threshold)$ satisfies
$F(\delta,\threshold,r(\delta,\threshold))=0$.
The density is given by
$$\density(\delta,r(\delta,\threshold)) = \frac{\pi\cdot r(\delta,\threshold)^2}{\delta},$$
which we want to maximize w.r.t.~$\delta$. This requires computing the derivative of $r(\delta,\threshold)$ w.r.t.~$\delta$. We do this by using the implicit function theorem
$$\frac{\partial r}{\partial\delta}(\delta,\threshold) = -\frac{\frac{\partial{F}}{\partial \delta} (\delta,\threshold,r)}
{\frac{\partial{F}}{\partial r} (\delta,\threshold,r)}.$$
After some calculations we find
\begin{equation*}
\frac{\partial\;\density(\delta,r(\delta,\threshold))}{\partial \delta} = \begin{cases}
\frac{\pi\sqrt{2r^2-\delta^2}}{2\delta\arccos\left(\frac{\delta}{r\sqrt{2}}\right)}& r_2\leq r \leq r_1,\; 0<\delta<\frac{1}{\sqrt{3}},\\ 
\frac{\pi(\delta^2-1)\sqrt{8r^2-\delta^2-1}}{8\delta^2\sqrt{\delta^2+1}\arccos\left(\frac{\delta^2+1}{8r^2}\right)}& r_1\leq r \leq r_2, \,\frac{1}{\sqrt{3}}<\delta<1,\\
\frac{\sqrt{8r^2-\delta^2-1}(\delta^2-1)+2\delta\sqrt{2r^2-\delta^2}\sqrt{\delta^2+1}}{4r\sqrt{\delta^2+1}\left(2\arccos\sqrt{\frac{\delta^2+1}{8r^2}}+\arccos\left(\frac{\delta}{r\sqrt{2}}\right)\right)}& \max(r_1,r_2)\leq r\leq r_3,\; 0<\delta<1.
\end{cases}
\label{eq_derivative}
\end{equation*}

One can easily check that the first derivative is non-negative for any $\delta$, except if $r$ equals the packing radius $r_2$ corresponding to $\threshold=0$, and we know the
optimal packing for this case. Similarly, the second derivative is non-positive except if $r$ equals the packing radius $r_1$. 
The third derivative is zero either if $\delta=\frac{1}{\sqrt{3}}$ or if $r$ equals the covering radius $r_3$ corresponding to $\threshold\geq\voloverlap(\delta,r_3)$, in which case the hexagonal lattice
is optimal as we argued above. Moreover, for $r<r_3$, the derivative is increasing for $\delta<{\frac{1}{\sqrt 3}}$ and decreasing for $\delta>{\frac{1}{\sqrt 3}}$. This concludes the proof.
\end{proof}

\subsection{Dimension 3}

\ignore{
The overlap $\voloverlap$ is given by the union of the spherical caps obtained by intersecting the neighboring balls with $V$ and can be computed using inclusion-exclusion:
\begin{equation}
\label{in_ex}
\voloverlap(\delta, r)= \frac{\sum_{\sigma\in\textrm{star(0)}}(-1)^{|\sigma|} \;\vol\left(V\cap\left(\cap_{j\in\sigma} B_r(j)\right)\right)}{\delta},
\end{equation}
where the sum is taken over all simplices $\sigma$ in the \emph{star} of the origin coming from the $\alpha$-shape with $\alpha = r$. Note that in dimension 3 we have $|\sigma|\leq 4$~(see \cite{union-balls}). Formulas for the intersection of one, two and three spherical caps have been described in \cite{edelsReport}. In order to understand for which radii we need to take into account the intersection of spherical caps, we need to do a similar analysis of the Voronoi cell V as in Section \ref{sec_dim_2}.
}

In three dimensions, the symmetry between $\delta$ and $\frac{1}{\delta}$ is lost, and we need to analyze both branches.
\ignore{
The analysis in dimension 2 was greatly simplified by the fact that all boundary vertices of the Voronoi cell $V$ have the same distance to the origin. Hence, a growing ball centered at the origin hits all vertices at the same time and we do not need to account for intersections of circular segments. In dimension 3, the same holds for $0<\delta\leq 1$, since the Delaunay triangulation consists of Freudenthal simplices only and they are all of the same kind. However, this is not true for $\delta\geq 1$, where the Delaunay triangulation equals the slice decomposition (see \cite[Sec.~2]{ek-covering}). This complicates the computations greatly.
}
We first discuss the case $0<\delta\leq 1$: Imagine that $r$ increases from 0 to $\infty$. Initially, the volume of the intersection $V\cap B_r$ equals the volume of the ball.
When increasing the ball radius $r$, there are three possibilities: 

\begin{enumerate}
\item[(i)] We hit a bisector plane. From now on we have to subtract 
a spherical cap from the volume. There are a total of 14 bisector planes of three different types. Their distance to the origin and number of occurrences are:

$$r_1:=\sqrt{\frac{\delta^2+2}{12}} \;\;\;\textrm{(6 planes)},\quad r_2:=\sqrt{\frac{2\delta^2+1}{6}} \;\;\;\textrm{(6 planes)},\quad r_3:=\frac{\delta\sqrt{3}}{2} \;\;\;\textrm{(2 planes)}.$$ 
\item[(ii)] We hit a boundary edge of $V$, where two bisector planes are meeting.
From now on, we have to add to the volume the intersection of the two spherical caps
involved (because they are counted twice). There are a total of 36 trisector edges of two different types. Their distance to the origin, number of occurrences, and types of bisector planes between the 3 involved spheres are:
$$r_4:=\frac{\delta^2+2}{3\sqrt{2}} \,\;\;\textrm{(18 edges of type 1-1-2)},\quad r_5:=\frac{\sqrt{(\delta^2+2)(2\delta^2+1)}}{2\sqrt{3}}  \,\;\;\textrm{(18 edges of type 1-2-3)}.$$
However, note that the volume of the cap intersection depends on the type of the bisector plane between the two spheres that are not centered at the origin. We get $5$ different subtypes,
four of them appearing $6$ times, and one appearing $12$ times in the polytope.

\item[(iii)] We hit a boundary vertex of $V$. All 24 boundary vertices have the same distance to the origin, namely the covering radius
$$r_6:=\frac{1}{6}\sqrt{8\delta^4+11\delta^2+8}.$$
When $r$ exceeds $r_6$ the whole Voronoi cell $V$ is covered, so the intersection
has volume $\delta$.
\end{enumerate}

Depending on the value of $\delta$ we have the following ordering of the critical radii:
\begin{eqnarray*}
r_3\leq r_1\leq r_2\leq r_5\leq r_4\leq r_6 \qquad &\textrm{for } 0\leq\delta\leq 1/2,\\
r_1\leq r_3\leq r_2\leq r_4\leq r_5\leq r_6 \qquad & \textrm{for } 1/2\leq\delta\leq \sqrt{2/5},\\
r_1\leq r_2\leq r_3\leq r_4\leq r_5\leq r_6 \qquad & \textrm{for } \sqrt{2/5}\leq\delta\leq \sqrt{\frac{19}{4}-\frac{3}{4}\sqrt{33}},\\
r_1\leq r_2\leq r_4\leq r_3\leq r_5\leq r_6 \qquad & \textrm{for } \sqrt{\frac{19}{4}-\frac{3}{4}\sqrt{33}}\leq\delta\leq 1.
\end{eqnarray*}

So $\voloverlap(\delta,r)$ seen as a function in $\delta$ has $4$ branches. In every branch, the interval which $r$ falls into determines how many and which types of cap intersections have to be taken into account to compute the volume-based overlap. 

For $\delta>1$, a similar analysis can be performed. However, there is one remarkable difference: The vertices of $V$ are no longer arranged in the same distance around the origin. More precisely,
there are $2$ vertices at distance $s_1$ and $6$ vertices at distance $s_2$ with
 $$s_1:=\frac{\delta^2+2}{2\sqrt{3}\delta}, \qquad s_2:=\frac{\sqrt{\delta^2+8}}{2\sqrt{3}}.$$
Note that $s_1<s_2$ and $s_2$ is the covering radius. So for $\delta>1$ and $s_1<r<s_2$ we need to take into account also triple intersections of spherical caps.

\begin{figure}
\centering
\includegraphics[width=5.5cm]{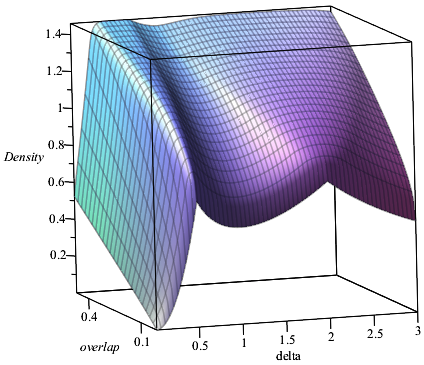}\;
\includegraphics[width=5.2cm]{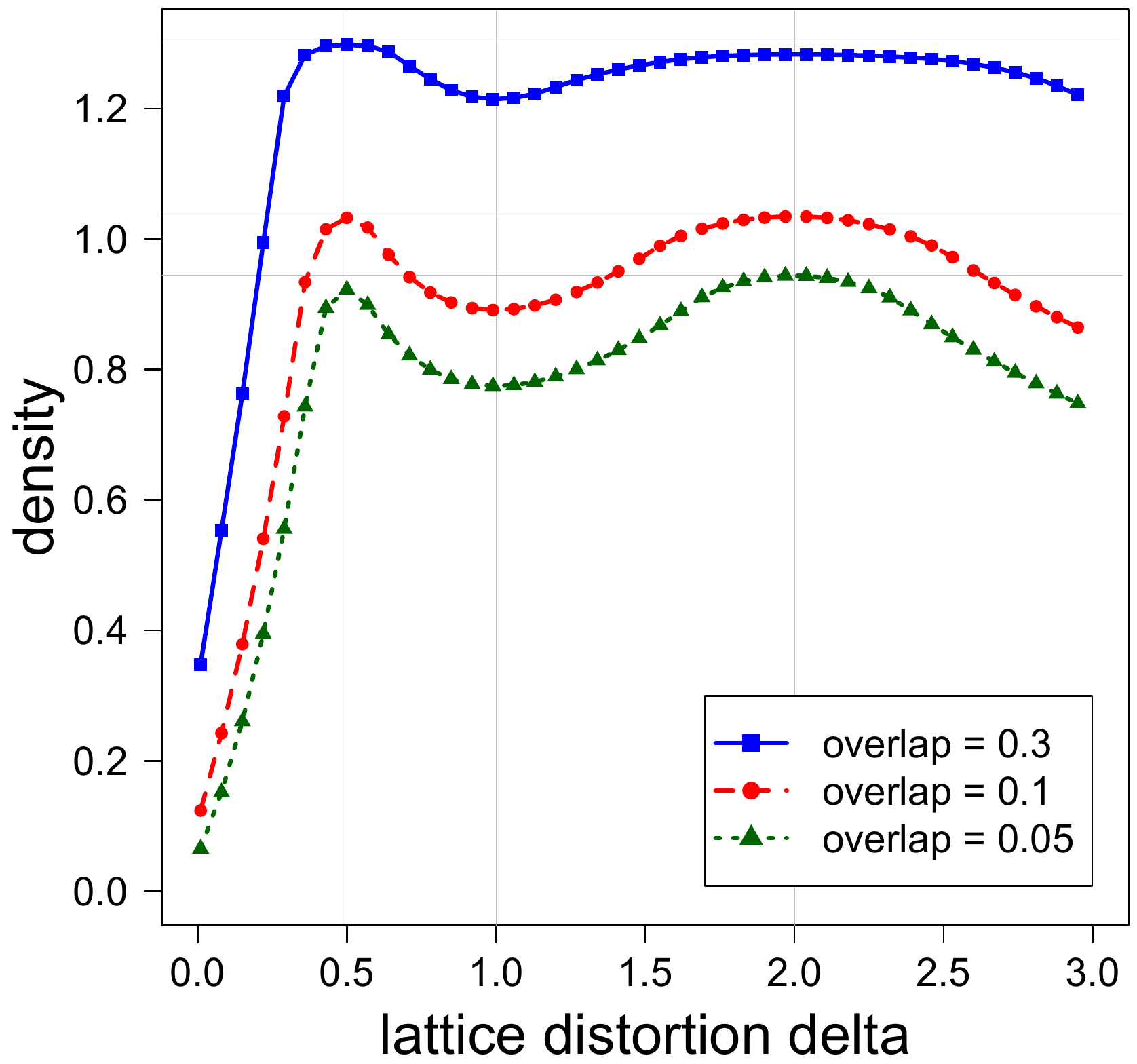}\;
\includegraphics[width=5.2cm]{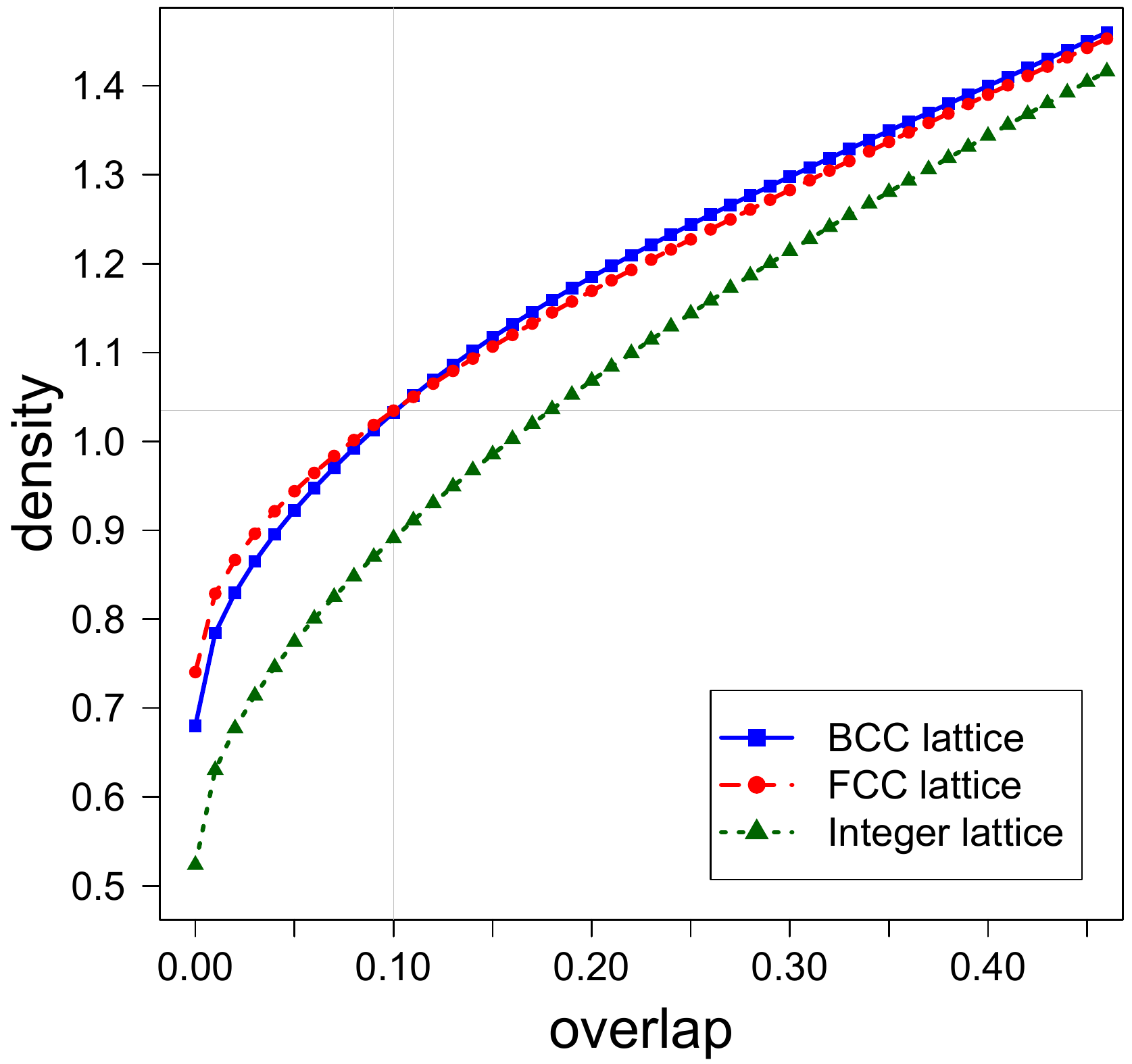}
\caption{Relaxed packing quality. Surface plot of $\packingquality(\delta, \omega)$ (left), slices through the surface $\packingquality$ for $\omega = 0.05$, 0.1 and 0.3 (middle), and slices through the surface $\packingquality$ for the BCC lattice, the integer lattice and the FCC lattice (right).}
\label{fig:3d}
\end{figure}

Formulas for the intersection of one, two and three spherical caps have been described in \cite{edelsReport}. In combination with our analysis, 
they result in a branchwise-defined closed expression for $\voloverlap(\delta,r)$. We have computed these expressions using the computer algebra system \textsc{MAPLE}.\footnote{\url{http://www.maplesoft.com}} A 3-dimensional plot of the function $\packingquality(\delta, \omega)$ is shown in Figure~\ref{fig:3d} (left). In Figure~\ref{fig:3d} (middle and right) we highlight specific slices through the 3-dimensional plot to better explain the behavior. Figure~\ref{fig:3d} (middle) shows $\packingquality(\delta, \omega)$ for three different values of $\threshold$. We can observe that the FCC lattice ($\delta=2$) is indeed optimal for small values 
of allowed overlap $\threshold$. When $\threshold = 0.1$, the BCC ($\delta=0.5$) and the FCC lattice achieve approximately the same density, namely $\density = 1.03$. Interestingly, for larger values of $\threshold$ the BCC lattice
attains the maximal density and surpasses the FCC lattice. Also observe that both lattices always achieve a better relaxed packing quality than the integer lattice ($\delta = 1$). 
Looking at the density of the FCC and the BCC lattice depending on $\threshold$ in Figure~\ref{fig:3d} (right), we can note that there is indeed only
one switch of optimality (at $\threshold \approx 0.1$).

Our analysis indicates that the $FCC$ and the $BCC$ lattice are always locally optimal configurations, and no other lattice from the family yields a better packing, independent of the allowed overlap. The natural next step would be to prove our observation. This problem can in theory be tackled with the same approach that we used in Section~\ref{sec_dim_2} in the $2D$ case by 
relating the derivative of $\packingquality(\delta,\omega)$ to the partial derivatives of $\voloverlap$ using the implicit function theorem. For small values of $\threshold$, we were able to verify the claim, that is, prove
monotonicity of the function in all branches with a substantial amount of symbolic computations. However, as soon as the expression for $\voloverlap$ involves intersections of $2$ and $3$ spherical caps, the derivatives
seem to become too complicated to be handled analytically. 

\section{Discussion}
\label{discussion}
This work has analyzed the problem of densest sphere packings while allowing some overlap among the spheres. We see our contributions as a first step towards an interesting and important research direction,
given the numerous applications of spheres with overlap in the natural sciences. For example, our analysis of the distance-based overlap measure showing that the FCC lattice is optimal independent of the amount of overlap, and hence independent of the scaling of the spheres, lays the theoretical foundations for \cite{ellipsoid_packing}, i.e.,~for analyzing the spatial organization of chromosomes in the cell nucleus as a sphere arrangement. A major restriction of our approach is our focus on a one-dimensional sub-lattice, the diagonally distorted
lattices. Can we hope for an analysis of more general lattice families? This question should probably first be considered
in $2D$, given the extremely involved proof of optimality already for the classical packing problem in $3D$.


\begin{thebibliography}{10}
\newcommand{\enquote}[1]{``#1''}
\providecommand{\url}[1]{\texttt{#1}}
\providecommand{\urlprefix}{URL }

\bibitem{Bambah1954}
R.~P. Bambah: \enquote{On lattice coverings by spheres}.
\newblock \emph{Proceedings of the National Institute of Sciences of India}
  \textbf{20} (1954) 25--52.

\bibitem{cs-sphere}
J.~Conway, N.~Sloane: \emph{Sphere Packings, Lattices and Groups}.
\newblock Springer, 3rd edn., 1999.

\bibitem{chromosome_territories}
T.~Cremer, M.~Cremer: \enquote{Chromosome territories}.
\newblock \emph{Cold Spring Harbor Perspectives in Biology}  (2010).

\bibitem{rfield_packing_1}
S.~H. DeVries, D.~A. Baylor: \enquote{Mosaic arrangement of ganglion cell
  receptive fields in rabbit retina}.
\newblock \emph{Journal of Neurophysiology} \textbf{78} (1997) 2048--2060.

\bibitem{edelsReport}
H.~Edelsbrunner, P.~Fu: \emph{Measuring space filling diagrams and voids}.
\newblock Tech. Rep. UIUC-BI-MB-94-01, Molecular Biophysics Group, Beckman
  Institute, University of Illinois at Urbana-Champaign, Illinois, 1994.

\bibitem{ek-covering}
H.~Edelsbrunner, M.~Kerber: \enquote{Covering and packing with spheres by
  diagonal distortion in $R^n$}.
\newblock In: G.~R. C.~Calude, A.~Salomaa (eds.) \emph{Rainbow of Computer
  Science - Essays Dedicated to Hermann Maurer on the Occasion of his 70th
  Birthday}, 20--35. Springer, 2011.

\bibitem{Hales2005}
T.~Hales: \enquote{A proof of the Kepler conjecture}.
\newblock \emph{Annals of Mathematics, Second Series} \textbf{162} (2005)
  1065--1185.

\bibitem{rfield_packing_2}
Y.~Karklin, E.~P. Simoncelli: \enquote{Efficient coding of natural images with
  a population of noisy linear-nonlinear neurons}.
\newblock \emph{Advances of Neural Information Processing Systems} \textbf{24}
  (2011) 999--1007.

\bibitem{hexPacking1939}
R.~Kershner: \enquote{The number of circles covering a set}.
\newblock \emph{American Journal of Mathematics} \textbf{61} (1939) 665--671.

\bibitem{neuron_packing_1}
A.~Raj, Y.~Chen: \enquote{The Wiring Economy Principle: Connectivity Determines
  Anatomy in the Human Brain}.
\newblock \emph{PLoS ONE} \textbf{6} (2011).

\bibitem{neuron_packing_2}
M.~Rivera-Alba, S.~N. Vitaladevuni, Y.~Mishchenko, Z.~Lu, S.~Takemura,
  L.~Scheffer, I.~A. Meinertzhagen, D.~B. Chklovskii, G.~G. de~Polavieja:
  \enquote{Wiring Economy and Volume Exclusion Determine Neuronal Placement in
  the Drosophila Brain}.
\newblock \emph{Current Biology} \textbf{21} (2011) 2000--2005.

\bibitem{ellipsoid_packing}
C.~Uhler, S.~J. Wright: \enquote{Packing ellipsoids with overlap}.
\newblock \emph{SIAM Review} \textbf{55} (2013) 671--706.

\end{thebibliography}
\end{document}